\documentclass{article}
\usepackage[utf8]{inputenc} 
\usepackage[T1]{fontenc}
\usepackage{url}              
\usepackage{cite}             

\usepackage[cmex10]{amsmath}  
\usepackage{mleftright}       
\mleftright                   

\usepackage{graphicx}         
\usepackage{booktabs}         
\usepackage{mathdots}

\usepackage{amssymb,amsfonts}
\usepackage{enumerate}
\usepackage{comment}
\usepackage{tabularx}
\usepackage{floatrow} 

\newtheorem{theorem}{Theorem}

\newtheorem{lemma}{Lemma}
\newtheorem{proposition}{Proposition}

\newtheorem{definition}{Definition}
\newtheorem{remark}{Remark}
\newtheorem{example}{Example}
\newtheorem{proof}{Proof}

\title{Binary convolutional codes with optimal column distances}
\author{Zita Abreu, Julia Lieb, Joachim Rosenthal}
\date{}

\begin{document}

\maketitle
\begin{abstract}
    There exists a large literature of construction of convolutional codes with maximal or near maximal free distance. Much less is known about constructions of convolutional codes having optimal or near optimal column distances.
  In this paper, a new construction of convolutional codes over the binary field with optimal column distances is presented. 
\end{abstract}

\section{Introduction}
Currently, all real communication channels are noisy, so there is a need for communication systems to use error-correcting codes. 
The distance of a code provides a measure for evaluating its ability to protect data from errors. Codes with larger distance are better because they allow to correct more errors.
One type of error-correcting codes is convolutional codes, which are very suitable for erasure channels, such as the Internet.
Convolutional codes possess different kinds of distance notions.
 One type of distance for convolutional codes is the column distances, which are considered for sequential decoding of the received information with low delay. Moreover, there is the free distance, which is considered when decoding delay does not matter and the decoding is only done after the codeword is fully received. However, the main advantage of convolutional codes is their suitability for sequential decoding with low delay.
 
In the past little progress has been made in finding good binary convolutional codes and so far optimal binary convolutional codes have only been presented for some special values of the code rate. There are two tabulations of binary convolutional codes with maximal free distance for rates $1/2$, $1/3$, $1/4$, $2/3$ and $3/4$; see \cite{Larsen,Paaske}. Moreover, in \cite{joh}, tables of binary convolutional codes of rates $1/2$ and $2/3$ with optimal column distances are presented.

In this paper, a construction of binary convolutional codes with optimal column distances for more general code rates will be presented and for that we focus on maximizing especially the small column distances that are most important for low delay decoding. In order to achieve such optimal constructions, we use a class of punctured simplex (block) codes, which we call partial simplex codes.

The paper is organized into three main parts. Section III provides upper and lower limits for column distances. Section IV presents the construction of binary convolutional codes of rate $1/n$ with optimal column distances and, finally, Section V completes the previous one, extending the presented construction to convolutional codes with dimension $k>1$.

\section{Preliminaries}


In this section, we present some definitions and results that are important for the following sections. For more details, we refer to e.g. \cite{LinCostello} or \cite{bookchapter}.
\begin{definition}
A \textbf{simplex code} $\mathcal{S}(k)$ of dimension $k$ is a block code $C=\{u\cdot S(k), u\in\mathbb F_2^k\}$ whose generator matrix $S(k)\in\mathbb F_2^{k\times (2^{k}-1)}$ has
all nonzero vectors in $\mathbb F_2^k$ as columns. 
\end{definition}

Note that $\mathcal{S}(k)$ is only unique up to column permutations inside the generator matrix leading to an equivalent code.

\begin{proposition}
All nonzero codewords of a $k$-dimensional simplex code of length $n=2^{k}-1$ have weight $2^{k-1}=\frac{n+1}{2}$.
\end{proposition}


\begin{definition}
A \textbf{convolutional code} $\mathcal{C}$ of rate $k/n$ is a $\mathbb{F}_{q}[z]$-submodule of $\mathbb{F}_{q}[z]^n$ of rank $k$, where $\mathbb{F}_{q}[z]$ is the ring of polynomials with coefficients in the field $\mathbb{F}_{q}$. A matrix $G(z)\in \mathbb{F}_{q}[z]^{k \times n}$ whose rows constitute a basis of $\mathcal{C}$ is called a \textbf{generator matrix} for $\mathcal{C}$, i.e.:
{\begin{eqnarray} 
\mathcal{C} &\hspace{-1.5mm}
=\hspace{-1.5mm}& \{{v(z) \in \mathbb{F}_{q}[z]^{n}: v(z) = u(z)G(z) \text{ with } u(z) \in \mathbb{F}_{q}[z]^{k}\}.}\nonumber
\end{eqnarray}}
\end{definition}

\begin{definition}
Let $G(z)=\sum_{i=0}^{\mu}G_i z^{i}\in\mathbb F_q[z]^{k\times n}$ with $G_{\mu} \neq 0$ and $k\leq n$. For each $i$, $1\leq i\leq k$, the $i$-th \textbf{row degree} $\nu_i$ of $G(z)$ is defined as the largest degree of any entry in row $i$ of $G(z)$, in particular $\mu=\max_{i=1,\hdots,k}\nu_i$. The \textbf{external degree} of $G(z)$ is the sum of the row degrees of $G(z)$. The \textbf{internal degree} of $G(z)$ is the maximal degree of the  $k\times k$ minors of $G(z)$.
\end{definition}

\begin{definition}
A matrix $G(z)\in \mathbb F_q[z]^{k\times n}$ is said to be
\textbf{row reduced}\index{row reduced} if its internal and external degrees are equal. In this case, $G(z)$ is called a minimal generator matrix of the convolutional code it generates. The \textbf{degree} $\delta$ of a code $\mathcal{C}$ is
the external degree of a minimal generator matrix of $\mathcal{C}$. A convolutional code with rate $k/n$ and degree $\delta$ is called an $(n,k,\delta)$ convolutional code.
\end{definition}

\begin{definition}
$G(z)\in\mathbb F_q[z]^{k\times n}$ is said to have \textbf{generic row degrees}
if $\nu_1=\hdots=\nu_t=\lceil\frac{\delta}{k}\rceil$ and $\nu_t=\hdots=\nu_k=\lfloor\frac{\delta}{k}\rfloor$ for $t=\delta+k-k\lceil\frac{\delta}{k}\rceil$.
\end{definition}

\begin{definition}
A generator matrix $G(z)\in {\mathbb{F}_{q}[z]}^{k\times n}$ with $G_0=G(0)$ full (row) rank is called \textbf{delay-free}. 
\end{definition} 

\begin{definition}
The (Hamming) weight
of a polynomial vector $v(z)=\sum_{t=0}^{\deg(v(z))}v_tz^t \in \mathbb{F}_q[z]^n$ is defined as $wt(v(z))=\sum_{t=0}^{\deg(v(z))}wt(v_t)$, where $wt(v_t)$ is the weight of $v_t\in\mathbb F_q^n$. 
\end{definition}

\begin{definition}
{\em
The \textbf{free distance}\index{code!minimum distance}\index{distance!free} of a convolutional code
  $\mathcal{C}$ is given by
  \[d_{free}(\mathcal{C}):=\min_{v(z)\in\mathcal{C}}\left\{wt(v(z))\
    |\ v(z) \neq 0\right\}.\]
    }
\end{definition}

\begin{definition}
Let $G(z)=\sum_{i=0}^{\mu}G_iz^i\in \mathbb{F}_q[z]^{k\times n}$ be a generator matrix of a convolutional code $\mathcal{C}$. For $j\in\mathbb N_0$, define the \textbf{truncated sliding
generator matrices} as
\begin{align*}
  G_j^c:=\arraycolsep=5pt
  \left[
  \begin{array}{ccc}
    \hspace{-1.5mm} G_0 \hspace{-1.5mm} & \hspace{-1.5mm}  \hspace{-1.5mm} \cdots \hspace{-1.5mm}& \hspace{-1.5mm} G_j \hspace{-1.5mm}\\
         & \hspace{-1.5mm} \ddots \hspace{-1.5mm} \hspace{-1.5mm}& \hspace{-1.5mm} \vdots \hspace{-1.5mm} \hspace{-1.5mm}\\
       & \hspace{-1.5mm} & \hspace{-1.5mm} G_0 \hspace{-1.5mm}
  \end{array}
              \right]\in \mathbb F_q^{(j+1)k\times (j+1)n}
              \end{align*}
where we set $G_i=0$ for $i>\mu$.
\end{definition}

\begin{definition}
For $j\in\mathbb N_0$, the \textbf{j-th column distance}\index{column distance} of a convolutional code $\mathcal{C}$ is defined as
\[
d_j^c(\mathcal{C}):=\min\left\{wt(v_0,\hdots,v_j)\ |\ {v}(z)\in\mathcal{C} \text{ and }{v}_0 \neq 0\right\}.
\]
\end{definition}

Since the convolutional codes which we will construct in this paper will all be delay-free, we can use that in this case 
\[
\small
d_j^c(\mathcal{C})=\min\left\{wt(u_0,\hdots,u_j)G_j^c\ |\ u(z)G(z)\in\mathcal{C} \text{ and }{u}_0 \neq 0\right\}
\]
\normalsize

\begin{definition}\label{Def:Catas-and-Non-catas}
Let $\mathcal{C}$ be an $(n,k)$ convolutional code over $\mathbb{F}_q$. A full row rank matrix $H(z) \in \mathbb{F}_q[z]^{(n-k)\times n}$ satisfying
\[
\mathcal{C} = \ker H(z) = \{v(z) \in \mathbb{F}_q[z]^n \, : \, H(z)v(z)^\top = 0\}
\]
is called a \textbf{parity-check matrix} of $\mathcal{C}$. If such a matrix exists, $\mathcal{C}$ is called \textbf{non-catastrophic}, otherwise it is called \textbf{catastrophic}.
\end{definition}

A code is non-catastrophic if and only if $G(z)$ is left prime which is equivalent to $G(z)$ having full row rank for all elements from the closure $z\in\overline{\mathbb F}_q$ \cite{bookchapter}. This implies that each non-catastrophic convolutional code is delay-free.
Moreover, if $\mathcal{C}$ is non-catastrophic, $d_{free}(\mathcal{C})= \lim_{j\rightarrow\infty} d_j^c (\mathcal{C})$, what can be used to calculate the free distance of the constructions and examples in the following sections.
\begin{theorem}[{\cite{ro99a1}, \cite{gl06}}]\label{th2} Let $\mathcal{C}$ be an $(n,k,\delta)$ convolutional code. Then,
\begin{itemize}
\item[(i)] $d_{free}(\mathcal{C})\leq(n-k)\left(\left\lfloor
    \frac{\delta}{k}\right\rfloor+1\right)+\delta+1$ 
\item[(ii)] $d_j^c (\mathcal{C}) \leq (n-k)(j + 1) + 1$ for all
    $j\in\mathbb N_0$
\end{itemize}
\label{ub}
\end{theorem}

The bound in (i) of Theorem \ref{th2} is called \textbf{generalized Singleton bound}. The fact that $d_j^c(\mathcal{C})\leq d_{free}(\mathcal{C})$ for all $j\in\mathbb N_0$ implies
$d_j^c (\mathcal{C})
\leq(n-k)\left(\left\lfloor\frac{\delta}{k}\right\rfloor+1\right)+\delta+1$
for all $j\in\mathbb N_0$. Hence
$j=L:=\left\lfloor\frac{\delta}{k}\right\rfloor+\left\lfloor\frac{\delta}{n-k}\right\rfloor$
is the largest possible value of $j$ for which $d_j^c(\mathcal{C})$ can
attain the upper bound in (ii) in Theorem \ref{th2}. Moreover, the next lemma shows that if $d_j^c(\mathcal{C})$ is maximal, then the same holds for $d_i^c(\mathcal{C})$ for all $i\leq j$. 

\begin{lemma}[\cite{gl06}]
Let $\mathcal{C}$ be an $(n,k,\delta)$ convolutional code.  If $d_j^c(\mathcal{C})=(n-k)(j+1)+1$ for some
  $j\in\{1,\hdots,L\}$, then  $d_i^c(\mathcal{C})=(n-k)(i+1)+1$ for all
  $i\leq j$.
\label{maxim}
\end{lemma}

This leads to the following definition.

\begin{definition}[\cite{gl06}]
An $(n,k,\delta)$ convolutional code $\mathcal{C}$ is said to be \textbf{maximum distance profile (MDP)}\index{maximum distance profile (MDP) code} if 
 \[
 d_j^c(\mathcal{C})=(n-k)(j+1)+1\ \text{for} \ j=L=\left\lfloor\frac{\delta}{k}\right\rfloor+\left\lfloor\frac{\delta}{n-k}\right\rfloor.
\]
\end{definition}

It is known that for the existence of MDP codes the size of the underlying finite field has to be sufficiently large (see e.g. \cite{hutch,necfs}), i.e. we cannot construct MDP codes over the binary field. In the following, we investigate upper bounds on the column distances that can be achieved by binary convolutional codes and how to obtain constructions for binary convolutional codes with optimal column distances.

\section{Upper and lower bounds for column distances}

In this section, we present some bounds on the column distances of convolutional codes that will be helpful to show that the constructions we will present in the following sections are optimal convolutional codes.

\begin{lemma}\label{bounds}
Let $\mathcal{C}$ be an $(n,k,\delta)$ convolutional code with generator matrix $G(z)=\sum_{i=0}^{\mu}G_i z^{i}\in\mathbb F_q[z]^{k\times n}$ with $G_{\mu} \neq 0$. Denote by $wt_r(G_i)$ the weight of row $r$ of $G_i$. Then, 
\begin{equation}\label{main1}
\small
\begin{aligned}
  \sum_{i=0}^j \min_{u_0\neq 0} wt\left((u_0\ \hspace{-1.5mm} \cdots\ \hspace{-1.5mm} u_i)\begin{pmatrix}
    G_i\vspace{-1.5mm}\\ \vdots\\ \vspace{-1.5mm}G_0 \end{pmatrix}\right)\hspace{-0.5mm}  \leq d_j^c(\mathcal{C})
  \leq \hspace{-1.5mm} \min_{r\in\{1,\hdots,k\}}\hspace{-1.5mm}\sum_{i=0}^{\min(j,\delta)}\hspace{-1.5mm}wt_r(G_i) 
\end{aligned}
\normalsize
\end{equation}
\begin{equation}\label{main2}
  \text{and}\quad \min_{r\in\{1,\hdots,k\}}\hspace{-1.5mm}\sum_{i=0}^{\min(j,\delta)}\hspace{-1.5mm}wt_r(G_i) \leq n((\min(j,\delta)+1).  \end{equation}
\end{lemma}

\begin{proof}
\eqref{main2} is obvious. For \eqref{main1} recall that by definition 
\begin{equation}\label{main}
\small
\begin{aligned}
    d_j^c(\mathcal{C})
         =\min_{u_0\neq 0}\sum_{i=0}^j wt\left((u_0\ \cdots\ u_i)\begin{pmatrix}
    G_i\vspace{-1.5mm}\\ \vdots\\ \vspace{-1.5mm}G_0 \nonumber
\end{pmatrix}\right).
\end{aligned}
\end{equation}
\normalsize From this the lower bound on $d_j^c(\mathcal{C})$ is clear. The upper bound follows as $d_j^c(\mathcal{C})$ is upper bounded by the weight of any of the first $k$ rows of $G_j^c$.
\end{proof}

These bounds are valid over any finite field, however from now on we will always be referring to the field $\mathbb{F}_2$. Later, we will see that these bounds can be reached with binary convolutional codes (in contrast to the bounds of Theorem \ref{th2}).

In the following sections, we will construct binary convolutional codes which are optimal in the following sense:

\begin{definition}\label{defopt}
   We say that a binary $(n,k,\delta)$ convolutional code $\mathcal{C}$ has \textbf{optimal column distances} if there exists no binary $(n,k,\delta)$ convolutional code $\hat{\mathcal{C}}$ such that $d^c_j(\hat{\mathcal{C}})>d^c_j(\mathcal{C})$ for some $j\in\mathbb N_0$ and $d^c_i(\hat{\mathcal{C}})=d^c_i(\mathcal{C})$ for all $0\leq i<j$.
\end{definition}

\section{Construction of rate $1/n$ binary convolutional codes with optimal column distances}


Firstly, we need to maximize $d_0^c$, i.e. we have to choose $G_0=(1\ \cdots\ 1)$.
The idea of the construction is to start with the generator matrix of a simplex code but only take the columns whose first entry is equal to 1 and set the resulting matrix equal to $\begin{pmatrix}
   G_0^\top \hspace{-1.5mm} & \hspace{-1.5mm} \cdots & \hspace{-1.5mm}  G_{\delta}^\top
\end{pmatrix}^\top$.

\begin{definition}\label{ps}
Take a generator matrix $S(\delta+1)$ of a simplex code 
and remove the columns with first entry equal to zero and define the resulting matrix as $S(\delta+1)_1\in\mathbb F_2^{(\delta+1)\times 2^{\delta}}$. For $m\in\mathbb N$, we call the (block) code with generator matrix $S(\delta+1)_1^m:=[S(\delta+1)_1\ \cdots\ S(\delta+1)_1]\in\mathbb F_2^{(\delta+1)\times m\cdot 2^{\delta}}$, an \textbf{$m$-fold partial simplex code} $\mathcal{S}(\delta+1)_1^m$ of dimension $\delta+1$. If $m=1$, we also just speak of partial simplex codes.
\end{definition}

\begin{proposition}
All codewords of $\mathcal{S}(\delta+1)_1^m$ except $(1 \cdots 1)\in\mathbb F_2^{m\cdot 2^{\delta}}$ 
have weight $m\cdot 2^{\delta-1}$. In particular, the minimum distance of such a code is equal to $m\cdot 2^{\delta-1}$.
\end{proposition}

\begin{proof}
First observe that it is enough to show the statement for $m=1$. Take a generator matrix $S(\delta+1)$ of a simplex code such that the first $2^{\delta}$ columns have a $1$ in the first row, i.e. write 
$S(\delta+1)=\left(S(\delta+1)_1\ \ \begin{array}{c}
     0_{1\times(2^{\delta}-1)}\\  S(\delta)
\end{array}\right)$. Since all codewords of $\mathcal{S}(\delta+1)$ have weight $2^{\delta}$ and all codewords in $\mathcal{S}(\delta)$ have weight $2^{\delta-1}$, all codewords of $\mathcal{S}(\delta+1)_1$ except the first row of $S(\delta+1)_1$ have weight $2^{\delta}-2^{\delta-1}=2^{\delta-1}$.
\end{proof}

Next, we construct binary convolutional codes with optimal column distances from $m$-fold partial simplex codes.

\begin{theorem}
Let $n=m\cdot 2^{\delta}$ and $\mathcal{C}$ be the $(n,1,\delta)$ convolutional code with generator matrix $G(z)=\sum_{i=0}^{\delta}G_i\in\mathbb F_2[z]^{1\times m\cdot 2^{\delta}}$ where $\begin{pmatrix}
   G_0^\top \hspace{-1.5mm} & \hspace{-1.5mm} \cdots & \hspace{-1.5mm}  G_{\delta}^\top
\end{pmatrix}^\top=S(\delta+1)^m_1$. Then, $\mathcal{C}$ is non-catastrophic and
$$d_j^c(\mathcal{C})=\begin{cases}
n+j\frac{n}{2} & \text{for}\quad j\leq\delta\\
n+\delta\frac{n}{2}  & \text{for}\quad j\geq\delta
\end{cases}\quad\text{and}\ d_{free}(\mathcal{C})=n+\delta\frac{n}{2}.$$
\end{theorem}

\begin{proof}
First, $\mathcal{C}$ is non-catastrophic since one of the entries of $G(z)$ is equal to $1$, because the first standard basis vector $e_1$ corresponds to one column of $S(\delta+1)^m_1$.

Obviously, $d_0^c(\mathcal{C})=n$. To calculate the remaining column distances, we apply Lemma \ref{bounds} and distinguish the two cases $j\leq\delta$ and $j\geq \delta$.
For $j=1,\hdots,\delta$, one has
\begin{equation}\label{main}
\begin{aligned}
\small
    d_j^c(\mathcal{C})-d_{j-1}^c(\mathcal{C}) & \geq \min_{u_0\neq 0}  wt\left((u_0\ \hspace{-1.5mm} \cdots\ \hspace{-1.5mm} u_j)\begin{pmatrix}
    G_j\vspace{-1.5mm}\\ \vdots\\ \vspace{-1.5mm}G_0
\end{pmatrix}\right)\\
         & =m \cdot 2^{\delta-j}\cdot 2^{j-1}=m\cdot 2^{\delta-1}=\frac{n}{2} \nonumber
\end{aligned}
\end{equation} \normalsize since $\begin{pmatrix}
    G_j^\top & \hspace{-1.5mm} \cdots & \hspace{-1.5mm} G_0^\top
\end{pmatrix}^\top$ is a generator matrix of a $(2^{\delta-j}\cdot m)$ - fold partial simplex code of dimension $j+1$ (and the condition $u_0\neq 0$ ensures that we do not obtain the codeword $(1\ \cdots \ 1)$ inside this partial simplex code). This shows $d_j^c(\mathcal{C})\geq n+j\frac{n}{2}$ for $j\leq\delta$. Moreover, the upper bound of \eqref{main1} yields that $d_j^c(\mathcal{C})\leq  \sum_{i=0}^{\min(j,\delta)}wt(G_i)  =n+j\frac{n}{2}$ for $j\leq\delta$.\\
For $j\geq\delta$, we obtain $n+\delta\frac{n}{2}=d_{\delta}^c(\mathcal{C})\leq d_j^c(\mathcal{C})\leq \sum_{i=0}^{\delta}wt(G_i)=n+\delta\frac{n}{2}$. 
\end{proof}

\begin{theorem}\label{optimal}
Let $\mathcal{C}$ be a binary $(m\cdot 2^{\delta},1,\delta)$ convolutional code constructed as in the previous theorem. Then, $\mathcal{C}$ has optimal column distances in the 
sense of Definition \ref{defopt}.
\end{theorem}

\begin{proof}
First observe that in our construction for any $j\in\{0,\hdots,\delta\}$, we have that $wt\left((u_0 \hspace{-1.5mm} \ \cdots\ \hspace{-1.5mm} u_j)\begin{pmatrix}
    G_j\vspace{-0.5mm}\\ \vdots\\ \vspace{-0.5mm}G_0
\end{pmatrix}\right)$ has the same value for all $(u_0\hspace{-1.5mm} \ \cdots\ \hspace{-1.5mm}u_j)$ with $u_0\neq 0$. Therefore, the lower bound in \eqref{main1} is sharp. Moreover, we saw in the previous proof that also the upper bound of \eqref{main1} is sharp. Hence, to achieve better column distances than with our construction, one would need to increase the weight of at least one $G_i$.\\ We use this to show via induction with respect to $j$ that our construction for the $G_j$ leads to optimal column distances. Obviously, the choice $G_0=(1\ \cdots\ 1)$ leads to optimal $d_0^c$ and then any $G_1$ with weight $\frac{n}{2}$ leads to optimal $d_1^c$ (column permutations do not change distances). Assume that for any $j\in\mathbb N$, our construction leads to optimal $d_j^c$ and we want to show that it also leads to optimal $d_{j+1}^c$. Therefore, we suppose that $G_0,\hdots, G_j$ are given as in the previous theorem and we need to find $G_{j+1}$ such that $d_{j+1}^c$ is optimal. For this we can assume $j<\delta$ as $G_i=0$ for $i>\delta$. As observed before, one can only get larger $d_{j+1}^c$ than in our construction if the weight of $G_{j+1}$ is larger than in our construction, i.e. larger than $\frac{n}{2}$. But then the weight of the sum of the first and the last row of $G_{j+1}^c$ is equal to $n+j\frac{n}{2}+x$ where $x$ is the weight of the sum of $G_{j+1}$ and $(1\ \cdots\ 1)$, which is smaller than $\frac{n}{2}$, i.e. we obtain a smaller $d^c_{j+1}$ than with our construction.
\end{proof}
In the following, we will extend this construction idea to $(n,1,\delta)$ convolutional codes where $n$ is not of the form $m\cdot 2^{\delta}$ for some $m\in\mathbb N$. 
For this, we use that if we keep the length $n$ and increase the degree from $\delta$ to $\delta+1$, the coefficient matrices of the generator matrix of the optimal code of degree $\delta+1$ have to coincide until $G_{\delta}$ with some optimal code of degree $\delta$. Similarly, if we keep the degree $\delta$ and increase the length from $n$ to $n+1$, the generator matrix for an optimal code of length $n+1$ has to coincide in its first $n$ entries with an optimal code of length $n$. 
Hence, we can use $S(\delta + 1)_{1}^{m}$ with $m = \lfloor\frac{n}{2^{\delta}}\rfloor$ 
for the construction and add $n-\lfloor\frac{n}{2^{\delta}}\rfloor\cdot 2^{\delta}$ further columns of $S(\delta+1)_1$. When deciding which additional columns of $S(\delta+1)_1$ to add, we just need to make sure to maximize the part of the weight produced by the additional columns, since we saw before that $wt\left((u_0\ \cdots\ u_j)\begin{pmatrix}
    G_j\vspace{-1.5mm}\\ \vdots\\ \vspace{-1.5mm}G_0
\end{pmatrix}\right)=\frac{n}{2}$ for all $(u_0\ \cdots\ u_j)$ with $u_0\neq 0$ if $\begin{pmatrix}
    G_j^\top & \hspace{-1.5mm} \cdots & \hspace{-1.5mm}G_0^\top
\end{pmatrix}^\top$ consists of rows of a generator matrix of a partial simplex code.
We will explain how to do this for small values of $\delta$ to illustrate the procedure.\\
For $\delta=k=1$, we know from the preceding results that in case $n$ is even, we obtain optimal column distances from $S(2)_{1}^{\frac{n}{2}}$. 
If $n$ is odd, to construct $\begin{pmatrix}
    G_0^\top &\hspace{-1.5mm} \cdots & \hspace{-1.5mm}G_{\delta}^\top
\end{pmatrix}^\top$ we can use $S(2)_{1}^{\lfloor\frac{n}{2}\rfloor}$ 
and add another column 
$S(2)_1=\begin{pmatrix}
    1 \hspace{-1.5mm} & \hspace{-1.5mm}1 \\ 1\hspace{-1.5mm}  & \hspace{-1.5mm}0 
\end{pmatrix}$. It is easy to see that in this case for the column distances it does not matter which of the two columns of $S(2)_1$ we choose and we obtain in any case that $d_0^c(\mathcal{C})=n$ and $d_{free}(\mathcal{C})=d_j^c(\mathcal{C})=n+\lfloor\frac{n}{2}\rfloor$ for $j\in\mathbb N$, which is optimal.\\
For $\delta=2$, we can use $S(3)_{1}^{\frac{n}{4}}$ 
in case $n\equiv 0 \mod 4$. If $n\not\equiv 0\mod 4$, to obtain $(n,1,\delta)$ convolutional codes $\mathcal{C}$ with optimal distances, we just need to find optimal $(s,1,\delta)$ convolutional codes $\mathcal{C}_{mod\ 4}$ with $s\in\{1,2,3\}$ such that $n\equiv s \mod 4$ to use it to extend $S(3)_{1}^{\lfloor\frac{n}{4}\rfloor}$.\\
For $s=1$, i.e. $n-1\equiv 0 \mod 4$, no matter which column of $S(3)_1=\begin{pmatrix}
    1 \hspace{-1.5mm} & \hspace{-1.5mm}1\hspace{-1.5mm} & \hspace{-1.5mm}1\hspace{-1.5mm} & \hspace{-1.5mm}1\\ 1\hspace{-1.5mm} & \hspace{-1.5mm}0\hspace{-1.5mm} & \hspace{-1.5mm}1\hspace{-1.5mm} & \hspace{-1.5mm}0\\ 1 \hspace{-1.5mm}& \hspace{-1.5mm}1\hspace{-1.5mm} & \hspace{-1.5mm}0\hspace{-1.5mm} & \hspace{-1.5mm}0
\end{pmatrix}$ we choose to construct $\begin{pmatrix}
    G_0^\top &\hspace{-1.5mm}  G_1^\top & \hspace{-1.5mm} G_2^\top
\end{pmatrix}^\top$, we obtain that $d_j^c(\mathcal{C}_{mod\ 4})=1$ for all $j\in\mathbb N_0$, i.e. $d_j^c(\mathcal{C})=n+\left(\frac{n-1}{2}\right)j$ for $j\leq\delta=2$ and $d_{free}(\mathcal{C})=d_j^c(\mathcal{C})=n+\left(\frac{n-1}{2}\right)\delta=2n-1$ for $j\geq\delta=2$.\\
For $s=2$, we know from the case $\delta=1$, which gives us $\begin{pmatrix}
    G_0^\top &\hspace{-1.5mm}  G_1^\top 
\end{pmatrix}^\top$, that to have optimal $d_0^c$ and $d_1^c$, we need to choose two columns of $S(3)_1$ of the form $\begin{pmatrix} 1 \hspace{-1.5mm}& \hspace{-1.5mm}1\\ 1\hspace{-1.5mm} & \hspace{-1.5mm}0\\ x\hspace{-1.5mm} & \hspace{-1.5mm}y\end{pmatrix}$ with $x,y\in\mathbb F_2$. Doing the calculations, one obtains (denoting $d_i^c=d_i^c(\mathcal{C}_{mod\ 4})$) in any case $d_0^c=2,\  d_1^c=3,\ d_2^c=3$. For $(x,y)\in\{(0,0), (0,1)\}$, $d_i^c=3$ for $i\geq 3$, for $(x,y)=(1,0)$, $ d_i^c=4$ for $i\geq 3$, for $(x,y)=(1,1)$, $d_3^c=d_4^c=4$ and $d_i^c=5$ for $i\geq 5$.
This means $(x,y)=(1,1)$ yields the unique optimal choice leading to $d_0^c(\mathcal{C})=n$, $d_1^c(\mathcal{C})=n+\frac{n}{2}$, $d_2^c(\mathcal{C})=2n-1$, $d_3^c(\mathcal{C})=d_4^c(\mathcal{C})=2n$, $d_{free}(\mathcal{C})=d_i^c(\mathcal{C})=2n+1$ for $i\geq 5$.\\
For $s=3$, using the previous results, we have the two options $\begin{pmatrix} 1\hspace{-1.5mm} & \hspace{-1.5mm}1\hspace{-1.5mm} & \hspace{-1.5mm}1\\ 1\hspace{-1.5mm} & \hspace{-1.5mm}0\hspace{-1.5mm} & \hspace{-1.5mm}1\\ 1\hspace{-1.5mm} & \hspace{-1.5mm}1\hspace{-1.5mm} & \hspace{-1.5mm}0\end{pmatrix}$ and $\begin{pmatrix} 1\hspace{-1.5mm} & \hspace{-1.5mm}1\hspace{-1.5mm} & \hspace{-1.5mm}1\\ 1\hspace{-1.5mm} & \hspace{-1.5mm}0\hspace{-1.5mm} & \hspace{-1.5mm}0\\ 1\hspace{-1.5mm} & \hspace{-1.5mm}1\hspace{-1.5mm} & \hspace{-1.5mm}0\end{pmatrix}$ for choosing three columns of $S(3)_{1}$. For $i\leq 3$, both lead to the same column distances $d_i^c=i+3$. But, for $i\geq 4$, the first option has $d_i^c=7$, while the second option has $d_i^c=6$. So, the first choice is optimal, leading to $d_0^c(\mathcal{C})=n$, $d_1^c(\mathcal{C})=n+\frac{n-1}{2}$, $d_2^c(\mathcal{C})=2n-1$, $d_3^c(\mathcal{C})=2n$, $d_{free}(\mathcal{C})=d_i^c(\mathcal{C})=2n+1$ for $i\geq 4$.

Doing all the previous calculations, one also observes that until $j=\delta$, the lower bound of \eqref{main1} is sharp and all choices for the new columns lead to the same $j$-th column distances. However, for $j> \delta$, different choices for the additional columns can lead to different $j$-th column distances, in which case the lower bound of \eqref{main1} is not sharp anymore. Also note that this lower bound can only increase until $j=\delta$, since for larger $j$ the new summands in the bound are zero because of $G_j=0$ for $j>\delta$. This implies that for fully optimizing all column distances it is not enough to just work with the lower bound of \eqref{main1} but to calculate the exact column distances. But it also means that focusing especially on the first column distances, maximizing just the lower bound in \eqref{main1} leads to very good results.
Since the computational effort is increasing with $\delta$, in the following we present codes with $\delta=3$ and $\delta=4$, where we optimized the lower bound in \eqref{main1} with the help of computer search.

For $\delta=3$, we need to find optimal $(s,1,3)$ convolutional codes $\mathcal{C}_s$ for $s=1,\hdots,7$. Let $G(z)$ be a generator matrix of $\mathcal{C}_s$. From the case $\delta=2$ we deduce that $\begin{pmatrix}
    G_0^\top &\hspace{-1.5mm}  G_1^\top & \hspace{-1.5mm} G_2^\top
\end{pmatrix}^\top$ has to be equal to the first $s$ columns of the matrix
$S(3)_1^2
.$
Denote by $wt^{s}$ the minimal weight of the code generated by the first $s$ columns of $\begin{pmatrix}
   S(3)_1^2\\  \tilde{G}_3 \end{pmatrix}$. Then, to maximize the lower bound of \eqref{main1}, we need to find $\tilde{G}_3$ such that $wt^s$ is maximized. Note that the additional condition $u_0\neq 0$ in \eqref{main1} does not matter because $u_0$ will be multiplied with $G_3$ and the already given first $3$ rows have been chosen in an optimal way.
We obtained with the help of the computer that in order to optimize $wt^s$, $\tilde{G}_3$ has to be equal to 
one of the following vectors:\\ 
$\left(\begin{array}{rrrrrrrr}
\hspace{-2mm} 0 \hspace{-2mm} & \hspace{-2mm} 0 \hspace{-2mm}& \hspace{-2mm} 0 \hspace{-2mm}& \hspace{-2mm} 1 \hspace{-2mm} &  \hspace{-2mm} 1 \hspace{-2mm}& \hspace{-2mm} 1 \hspace{-2mm}& \hspace{-2mm} 1 \hspace{-2mm}& \hspace{-2mm} 0 \hspace{-2mm}\end{array} \hspace{-2mm}\right)$, $\left(\begin{array}{rrrrrrrr}
\hspace{-2mm} 0 \hspace{-2mm} & \hspace{-2mm} 0 \hspace{-2mm}& \hspace{-2mm} 1 \hspace{-2mm}& \hspace{-2mm} 0 \hspace{-2mm} &  \hspace{-2mm} 1 \hspace{-2mm}& \hspace{-2mm} 1 \hspace{-2mm}& \hspace{-2mm} 0 \hspace{-2mm}& \hspace{-2mm} 1 \hspace{-2mm}\end{array}\hspace{-2mm}\right)$, $\left(\begin{array}{rrrrrrrr}
\hspace{-2mm} 0 \hspace{-2mm} & \hspace{-2mm} 1 \hspace{-2mm}& \hspace{-2mm} 0 \hspace{-2mm}& \hspace{-2mm} 0 \hspace{-2mm} &  \hspace{-2mm} 1 \hspace{-2mm}& \hspace{-2mm} 0 \hspace{-2mm}& \hspace{-2mm} 1 \hspace{-2mm}& \hspace{-2mm} 1 \hspace{-2mm}\end{array}\hspace{-2mm}\right)$,$\left(\begin{array}{rrrrrrrr}
\hspace{-2mm} 0 \hspace{-2mm} & \hspace{-2mm} 1 \hspace{-2mm}& \hspace{-2mm} 1 \hspace{-2mm}& \hspace{-2mm} 1 \hspace{-2mm} &  \hspace{-2mm} 1 \hspace{-2mm}& \hspace{-2mm} 0 \hspace{-2mm}& \hspace{-2mm} 0 \hspace{-2mm}& \hspace{-2mm} 0 \hspace{-2mm}\end{array}\hspace{-2mm}\right)$, 

$\left(\begin{array}{rrrrrrrr}
\hspace{-2mm} 1 \hspace{-2mm} & \hspace{-2mm} 0 \hspace{-2mm}& \hspace{-2mm} 0 \hspace{-2mm}& \hspace{-2mm} 0 \hspace{-2mm} &  \hspace{-2mm} 0 \hspace{-2mm}& \hspace{-2mm} 1 \hspace{-2mm}& \hspace{-2mm} 1 \hspace{-2mm}& \hspace{-2mm} 1 \hspace{-2mm}\end{array}\hspace{-2mm}\right)$, $\left(\begin{array}{rrrrrrrr}
\hspace{-2mm} 1 \hspace{-2mm} & \hspace{-2mm} 0 \hspace{-2mm}& \hspace{-2mm} 1 \hspace{-2mm}& \hspace{-2mm} 1 \hspace{-2mm} &  \hspace{-2mm} 0 \hspace{-2mm}& \hspace{-2mm} 1 \hspace{-2mm}& \hspace{-2mm} 0 \hspace{-2mm}& \hspace{-2mm} 0 \hspace{-2mm}\end{array}\hspace{-2mm}\right)$, $\left(\begin{array}{rrrrrrrr}
\hspace{-2mm} 1 \hspace{-2mm} & \hspace{-2mm} 1 \hspace{-2mm}& \hspace{-2mm} 0 \hspace{-2mm}& \hspace{-2mm} 1 \hspace{-2mm} &  \hspace{-2mm} 0 \hspace{-2mm}& \hspace{-2mm} 0 \hspace{-2mm}& \hspace{-2mm} 1 \hspace{-2mm}& \hspace{-2mm} 0 \hspace{-2mm}\end{array}\hspace{-2mm}\right)$, $\left(\begin{array}{rrrrrrrr}
\hspace{-2mm} 1 \hspace{-2mm} & \hspace{-2mm} 1 \hspace{-2mm}& \hspace{-2mm} 1 \hspace{-2mm}& \hspace{-2mm} 0 \hspace{-2mm} &  \hspace{-2mm} 0 \hspace{-2mm}& \hspace{-2mm} 0 \hspace{-2mm}& \hspace{-2mm} 0 \hspace{-2mm}& \hspace{-2mm} 1 \hspace{-2mm}\end{array}\hspace{-2mm}\right)$.

Since $S(3)_1^2$ is the generator matrix of a $2$-fold partial simplex code, it is clear that all $\tilde{G}_3$ have to be of the form\\ $(a\ \ b\ \ c\ \ d\ \ a+1\ \ b+1\ \ c+1\ \ d+1)$ with $a,b,c,d\in\mathbb F_2$ and that all of them have to have weight $4$. From the result of the computer search, we see in addition that $\tilde{G}_3$ is an optimal choice if and only if $a+b+c+d=1$. 
For all these $8$ optimal $\tilde{G}_3$, we obtain the following values for $wt^s$: 

\begin{center}\begin{tabular}{||l|c|r|r|r|r|r|r||}
    s & 1 & 2 & 3 & 4 & 5 & 6 & 7 \\
    $wt^{s}$ & 0 & 0 & 0 & 1 & 1 & 2 & 3 \\ 
\end{tabular}\end{center}
and we have that $d^c_3(\mathcal{C}_s)\geq d^c_2(\mathcal{C}_s)+wt_s$.

For $\delta=4$, we obtained with the help of the computer that for maximizing the lower bound of \eqref{main1} for $j=4$, it does not matter which of the eight options for $\tilde{G}_3$ to take. This means we can take any of these $\tilde{G}_3$ to form the matrix
$S(4)_1=\begin{pmatrix}
    S(3)_1^2 & S(3)_1^2\\
      \tilde{G}_3 & \tilde{G}_3
\end{pmatrix}$.
Denote now by $wt^t$ the minimal weight of the code generated by the first $t\in\{1,\hdots,15\}$ columns of $\begin{pmatrix}S(4)_1\\ \tilde{G}_4\end{pmatrix}$.
With the computer, we found that for each optimal choice $\tilde{G}_3=(\tilde{G}_3^1\ \ \tilde{G}_3^2)$ with $\tilde{G}_3^1, \tilde{G}_3^2\in\mathbb F_2^4$, there are the same eight optimal choices for $\tilde{G}_4$, namely 
exactly all vectors of the form $(\tilde{G}_3^1\ \  \tilde{G}_3^1\ \ \tilde{G}_3^2\ \ \tilde{G}_3^2)$.
In this way, we obtain 64 optimal codes
leading to the following optimal values for $wt^t$: 
\begin{center}
\begin{tabular}{||l|c|r|r|r|r|r|r|r|r|r|r|r|r|r|r||}
    \hspace{-2.5mm} t \hspace{-2.5mm} & \hspace{-2.5mm} 1 \hspace{-2.5mm} & \hspace{-2.5mm} 2 \hspace{-2.5mm} & \hspace{-2.5mm} 3 \hspace{-2.5mm} & \hspace{-2.5mm} 4 \hspace{-2.5mm} & \hspace{-2.5mm} 5 \hspace{-2.5mm} & \hspace{-2.5mm} 6 \hspace{-2.5mm} & \hspace{-2.5mm} 7 \hspace{-2.5mm} & \hspace{-2.5mm} 8 \hspace{-2.5mm} & \hspace{-2.5mm} 9 \hspace{-2.5mm} & \hspace{-2.5mm} 10 \hspace{-2.5mm} & \hspace{-2.5mm} 11 \hspace{-2.5mm} & \hspace{-2.5mm} 12 \hspace{-2.5mm} & \hspace{-2.5mm} 13 \hspace{-2.5mm} & \hspace{-2.5mm} 14 \hspace{-2.5mm} & \hspace{-2.5mm} 15 \hspace{-2.5mm}
    \\
    \hspace{-2.5mm} $wt^t$ \hspace{-2.5mm} & \hspace{-2.5mm} 0 \hspace{-2.5mm} & \hspace{-2.5mm} 0 \hspace{-2.5mm} & \hspace{-2.5mm} 0 \hspace{-2.5mm} & \hspace{-2.5mm} 0 \hspace{-2.5mm} & \hspace{-2.5mm} 1 \hspace{-2.5mm} & \hspace{-2.5mm} 1 \hspace{-2.5mm} & \hspace{-2.5mm} 1 \hspace{-2.5mm} & \hspace{-2.5mm} 2 \hspace{-2.5mm} & \hspace{-2.5mm} 2 \hspace{-2.5mm} & \hspace{-2.5mm} 3 \hspace{-2.5mm} & \hspace{-2.5mm} 4 \hspace{-2.5mm} & \hspace{-2.5mm} 4 \hspace{-2.5mm} & \hspace{-2.5mm} 5 \hspace{-2.5mm} & \hspace{-2.5mm} 6 \hspace{-2.5mm} & \hspace{-2.5mm} 7 \hspace{-2.5mm} 
    \\
\end{tabular}
\end{center}

If one has code parameters where $2^{\delta}\nmid n$, but does not want to do all the calculations described above to optimize the column distances, one can do the following.

\begin{theorem}
Set $m=\lfloor\frac{n}{2^{\delta}}\rfloor$, $n_1:=m\cdot 2^{\delta}$ and write $n-n_1=2^{a_1-1}+...+2^{a_b-1}$ with $b,a_i\in\mathbb N$ for $i=1,\hdots,b$ and $\delta \geq a_1>...>a_b$. Set $\begin{pmatrix}
  G_0^\top \hspace{-1.5mm} & \hspace{-1.5mm} \cdots & \hspace{-1.5mm}  G_{\delta}^\top
\end{pmatrix}^\top=[S(\delta+1)^m_1\ S]$
where $S$ consists of $n-n_1$ columns of $S(\delta+1)_1$ and has the form
$S=D_0=\begin{pmatrix}
    S(a_1)_1 & D_1\\
    \ast & \ast
\end{pmatrix}$, $D_1=\begin{pmatrix}
    S(a_2)_1 & D_2\\
    \ast & \ast
\end{pmatrix}$, ..., $D_i=\begin{pmatrix}
    S(a_{i+1})_1 & D_{i+1}\\
    \ast & \ast
\end{pmatrix}$, ..., $D_{b-1}=S(a_b)_1$.
Then, the $(n,1,\delta)$ binary convolutional code $\mathcal{C}$ with generator matrix $G(z)$ has column distances which are near optimal in the following sense:
For $j\leq a_b-1$, $d_j^c(\mathcal{C})=n+j\frac{n}{2}$, i.e. optimal, and for $a_{x+1}<j+1\leq a_x$ with $x\in\{1,\hdots,b-1\}$,
 \begin{align*}
     d_j^c(\mathcal{C})\geq \frac{n_1}{2}+2^{a_1-2}+\cdots+2^{a_x-2}+d^c_{j-1}(\mathcal{C})
 \end{align*}
 \end{theorem}
 
\begin{proof} One has
$d_j^c(\mathcal{C})\geq d^c_{j-1}(\mathcal{C})+\frac{n_1}{2}+\min_{u_0\neq 0}((u_j,\hdots,u_0)S_{j+1})$
   where $S_{j+1}$ is the matrix formed by the first $j+1$ rows of $S$.
   As the first $2^{a_1-1}+\cdots+2^{a_x-1}$ columns of $S_{j+1}$ form the generator matrix of a folded partial simplex code, one obtains $\min_{u_0\neq 0}((u_j,\hdots,u_0)S_{j+1})\geq 2^{a_1-2}+\cdots+2^{a_x-2}$.
\end{proof}
\begin{remark}
In \cite{joh}, some optimal binary convolutional codes are listed. The only parameters for which codes with optimal column distances are constructed in \cite{joh} and in this paper are $n=2$, $k=1$, $\delta\in\{1,2\}$. In these two cases, the optimal code is unique and clearly the same code is built in both papers. 
\end{remark}

\section{Construction of binary convolutional codes of dimension $k>1$ with optimal column distances}


If $G(z)=\sum_{i=0}^{\mu}G_iz^i\in\mathbb F_q[z]^{k\times n}$ with $\mu=\deg(G)$ is a row-reduced generator matrix for an $(n,k,\delta)$ convolutional code $\mathcal{C}$ with generic row degrees, then $\mu=\lceil\frac{\delta}{k}\rceil$ and $G_{\mu}$ has $k\lceil\frac{\delta}{k}\rceil-\delta$ zero rows and we will construct the code s.t. the last $k\lceil\frac{\delta}{k}\rceil-\delta$ rows of $G_{\mu}$ are zero.
  Denote by $\Tilde{G}_{\mu}\in\mathbb F^{(\delta+k-k\lceil\frac{\delta}{k}\rceil)\times n}$ the matrix consisting of the first $\delta+k-k\lceil\frac{\delta}{k}\rceil$, i.e. nonzero, rows of  $G_{\mu}$.
As $(G_0^\top\ \cdots\ \ G_{\mu-1}^\top \ \tilde{G}_{\mu}^\top)^\top$ has $  \delta+k$ rows
we will use simplex codes of dimension $\delta+k$ for the construction.

To obtain codes with optimal column distances, we need to start by choosing $G_0$ as generator matrix of an optimal binary (block) code, i.e. as generator matrix of a folded
simplex code of dimension $k$.
To optimize $d_1^c(\mathcal{C})$, choose $G_1$ such that $(G_0^T \ G_1^T)^{T}$ 
consists of part of the columns of a folded
simplex code. It cannot consist of all columns of such a code as there is no zero column in $G_0$, similar to the case $k=1$. This leads to the following generalization of Definition \ref{ps}.

\begin{definition}
Take a generator matrix $S(\delta+k)$
of a simplex code 
and remove the columns whose first $k$ entries are equal to zero and define the resulting matrix as $S(k+\delta)_k\in\mathbb F_2^{(\delta+k)\times (2^{\delta+k}-2^{\delta})}$. For $m\in\mathbb N$, we call the (block) code with generator matrix $S(\delta+k)_k^m:=[S(\delta+k)_k\ \cdots\ S(\delta+k)_k]\in\mathbb F_2^{(\delta+k)\times m\cdot (2^{\delta+k}-2^{\delta})}$ an \textbf{$m$-fold $k$-partial simplex code} $\mathcal{S}(\delta+k)_k^m$ of dimension $\delta+k$. 
\end{definition}

\begin{proposition}
All codewords of $\mathcal{S}(\delta+k)_k^m$
that are linear combinations of just the first $k$ rows of $S(\delta+k)_k^m$
have weight $m\cdot 2^{\delta+k-1}$ and all other codewords have weight $m\cdot (2^{\delta+k-1}-2^{\delta-1})$. In particular, the minimum distance of such a code is equal to $m\cdot (2^{\delta+k-1}-2^{\delta-1})=m\cdot 2^{\delta-1}(2^k-1)$.
\end{proposition}

\begin{proof}
Again it is enough to show the statement for $m=1$. Take a generator matrix $S(\delta+k)$
of a simplex code such that the first $k$ rows have zeros in the last $2^{\delta}-1$ entries, i.e. 
\begin{equation}\label{fullsimplex}
\small
    S(\delta+k)=\left(S(\delta+k)_k\ \ \begin{array}{c}
     0_{k\times(2^{\delta}-1)}\\  S(\delta)
\end{array}\right).
\end{equation}
\normalsize
Hence, linear combinations of the first $k$ rows of $S(\delta+k)_k$ have all weight $2^{\delta+k-1}$ and linear combinations involving one of the other rows have weight $ 2^{\delta+k-1}-2^{\delta-1}$.
\end{proof}

\begin{theorem}
Let $\mathcal{C}$ be an $(m\cdot 2^{\delta}(2^k-1),k,\delta)$ convolutional code with generator matrix $G(z)=\sum_{i=0}^{\lceil\frac{\delta}{k}\rceil}G_iz^i\in\mathbb F_2[z]^{k\times m\cdot 2^{\delta}(2^k-1)}$ where $(G_0^\top\ \cdots\ \ G_{\mu-1}^\top \ \tilde{G}_{\mu}^\top)^\top=S(\delta+k)^m_1$. Then, $\mathcal{C}$ is non-catastrophic and
$$d_j^c(\mathcal{C})=\begin{cases}
n\cdot\frac{2^{k-1}}{2^k-1}+j\frac{n}{2} & \text{for}\quad j\leq\lfloor\frac{\delta}{k}\rfloor\\
n\cdot\frac{2^{k-1}}{2^k-1}+\lfloor\frac{\delta}{k}\rfloor\cdot \frac{n}{2}  & \text{for}\quad j\geq\lfloor\frac{\delta}{k}\rfloor
\end{cases}$$
\end{theorem}

\begin{proof}
First, $\mathcal{C}$ is non-catastrophic since $I_k$ is a submatrix of $G(z)$.
Since $G_0$ is the generator matrix of an $\frac{n}{2^{k}-1}$-fold simplex code of dimension $k$, we obtain $d_0^c(\mathcal{C})=
n\cdot\frac{2^{k-1}}{2^k-1}$. 
For $j=1,\hdots,\lfloor\frac{\delta}{k}\rfloor$, one has
\begin{equation}
\small
\begin{aligned}
     \min_{u_0\neq 0}  wt\left((u_0\ \hspace{-1.5mm} \cdots\ \hspace{-1.5mm} u_j)\begin{pmatrix}
    G_j\vspace{-1.5mm}\\ \vdots\\ \vspace{-1.5mm}G_0
\end{pmatrix}\right) =\frac{n\cdot 2^{kj-1}(2^k-1)}{2^{kj}(2^k-1)} =\frac{n}{2} \nonumber
\end{aligned}
\end{equation}
\normalsize
since we have up to row permutations an $m$-fold $k$-partial simplex code $\mathcal{S}(k(j+1))^m_k$ with $m=\frac{n}{2^{kj}(2^k-1)}$. 
Moreover, the condition $u_0\neq 0$ ensures that we do not get a linear combination of the first $k$ rows of $S(k(j+1))^m_k$ and hence, all occurring weights are equal to the minimum weight and we obtain
$ d_j^c(\mathcal{C})-d_{j-1}^c(\mathcal{C})=\frac{n}{2}$.\\
For $j>\lfloor\frac{\delta}{k}\rfloor$, $G_j$ contains at least one zero row and the corresponding row in the sliding generator matrix has then weight $n\cdot\frac{2^{k-1}}{2^k-1}+\lfloor\frac{\delta}{k}\rfloor\cdot \frac{n}{2}$, i.e. the column distances cannot increase any further, due to the upper bound in \eqref{main1}.
\end{proof}

\begin{theorem}\label{optimal}
Let $\mathcal{C}$ be a binary $(m\cdot 2^{\delta}(2^k-1),k,\delta)$ convolutional code constructed as in the previous theorem. Then, $\mathcal{C}$ has optimal column distances in the sense of Definition \ref{defopt}. 
\end{theorem}

\begin{proof} We can assume $m=1$.
As for $k=1$, lower and upper bound of \eqref{main1} are sharp for our construction. Hence, to achieve better column distances than with our construction, one would need to increase the weight of at least one $G_i$. We use this to show via induction with respect to $j$ that our construction for the $G_j$ leads to optimal column distances. Obviously, the choice of $G_0$ leads to optimal $d_0^c$. 
Suppose that $G_0,\hdots, G_j$ are given as in the previous theorem and we need to find $G_{j+1}$ such that $d_{j+1}^c$ is optimal. For this we can assume $j<\lfloor\frac{\delta}{k}\rfloor$ as $G_i$ contains a zero row for $i>\lfloor\frac{\delta}{k}\rfloor$. One can only get larger $d_{j+1}^c$ than in our construction if the weight of all rows of $G_{j+1}$ is larger than in our construction. Assume we increase the weight of one row of $G_{j+1}$ as in our construction to obtain $\hat{G}_{j+1}$ and denote the number of the row to which it corresponds in $S(\delta+k)$, as in \eqref{fullsimplex}, by $r_1$ and the generator matrix with the increased weight in row $r_1$ by $\hat{S}(\delta+k)$. The first $k$ rows of $S(\delta+k)_{k}$ correspond to $S(k)$.
Hence, there exists one of the first $k$ rows of $\hat{S}(\delta+k)$, whose index we denote by $r_2$, such that the weight of the sum of rows $r_1$ and $r_2$ of $\hat{S}(\delta+k)$ is smaller than the minimum weight of $\mathcal{S}(\delta+k)$. But since we did not change anything in the last $2^{\delta}-1$ columns of $S(\delta+k)$, the weight decreased in the first $2^{\delta+k}-2^{\delta}$ columns, which correspond to $S(\delta+k)_k$ used to define our convolutional code. This shows the optimality of the column distances.
\end{proof}

\begin{example}
    Take $k=2$, $n=12$ and $\delta=2$, then $\mu=1$ and $\delta+k=4$. The optimal $G_0$, leading to $d_0^c=8$ is 
    $G_0= S(2)^{4}=\left(\hspace{-2mm}\begin{array}{cccccccccccc}
    1\hspace{-2mm} & \hspace{-2mm}1\hspace{-2mm} & \hspace{-2mm}0\hspace{-2mm} & \hspace{-2mm}1\hspace{-2mm} & \hspace{-2mm}1\hspace{-2mm} & \hspace{-2mm}0\hspace{-2mm} & \hspace{-2mm}1\hspace{-2mm} & \hspace{-2mm}1\hspace{-2mm} & \hspace{-2mm}0\hspace{-2mm} & \hspace{-2mm}1\hspace{-2mm} & \hspace{-2mm}1\hspace{-2mm} & \hspace{-1.5mm}0\\
    1\hspace{-1.5mm} & \hspace{-1.5mm}0\hspace{-1.5mm} & \hspace{-1.5mm}1\hspace{-1.5mm} & \hspace{-1.5mm}1\hspace{-1.5mm} & \hspace{-1.5mm}0\hspace{-1.5mm} & \hspace{-1.5mm}1\hspace{-1.5mm} & \hspace{-1.5mm}1\hspace{-1.5mm} & \hspace{-1.5mm}0\hspace{-1.5mm} & \hspace{-1.5mm}1\hspace{-1.5mm} & \hspace{-1.5mm}1\hspace{-1.5mm} & \hspace{-1.5mm}0\hspace{-1.5mm} & \hspace{-1.5mm}1
    \end{array}\hspace{-2mm} \right)$. To maximize $d_1^c$ we take $G_1=\left(\hspace{-2mm}\begin{array}{cccccccccccc}
    1\hspace{-2mm} & \hspace{-2mm}1\hspace{-2mm} & \hspace{-2mm}1\hspace{-2mm} & \hspace{-2mm}0\hspace{-2mm} & \hspace{-2mm}0\hspace{-2mm} & \hspace{-2mm}0\hspace{-2mm} & \hspace{-2mm}1\hspace{-2mm} & \hspace{-2mm}1\hspace{-2mm} & \hspace{-2mm}1\hspace{-2mm} & \hspace{-2mm}0\hspace{-2mm} & \hspace{-2mm}0\hspace{-2mm} & \hspace{-1.5mm}0\\
    1\hspace{-1.5mm} & \hspace{-1.5mm}1\hspace{-1.5mm} & \hspace{-1.5mm}1\hspace{-1.5mm} & \hspace{-1.5mm}1\hspace{-1.5mm} & \hspace{-1.5mm}1\hspace{-1.5mm} & \hspace{-1.5mm}1\hspace{-1.5mm} & \hspace{-1.5mm}0\hspace{-1.5mm} & \hspace{-1.5mm}0\hspace{-1.5mm} & \hspace{-1.5mm}0\hspace{-1.5mm} & \hspace{-1.5mm}0\hspace{-1.5mm} & \hspace{-1.5mm}0\hspace{-1.5mm} & \hspace{-1.5mm}0
    \end{array}\hspace{-2mm} \right)$ such that $(G_0^{T} \ G_1^T)^T=S(4)_2$
    and $d_1^c=14$.
\end{example}


To obtain convolutional codes with optimal column distances where $n$ is not of the form $m\cdot 2^{\delta}(2^k-1)$ for some $m\in\mathbb N$, we can do the same procedure as described at the end of the preceding section, i.e. use $\mathcal{S}(\delta+k)_k^m$ with $m=\lfloor\frac{n}{2^{\delta}(2^k-1)}\rfloor$ and add some columns of $S(\delta+k)_k$.

\section{Conclusion}
Convolutional codes with optimal or near optimal column distances are attractive as they are capable of correcting a maximal number of errors per time interval. In this paper, we start with simplex codes and using both the technique of puncturing and folding we are able to construct new binary convolutional codes whose column distances are optimal for certain parameters and near optimal for the other parameters.

\section{Acknowledgements}
This work is supported by the SNSF grant n. 188430 and by CIDMA through FCT, UIDB/04106/2020 and UIDP/04106/2020. The work of the first author was also supported by FCT grant UI/BD/151186/2021 and the work of the second author was also by Forschungskredit of the University of Zurich, grant no. FK-21-127.


\end{document}